\documentclass[submission]{eptcs}
\providecommand{\event}{WPTE 2016} 

\providecommand{\urlalt}[2]{\href{#1}{#2}}
\providecommand{\doi}[1]{doi:\urlalt{http://dx.doi.org/#1}{#1}}

\usepackage{amsfonts}
\usepackage{amsmath, amsthm}
\usepackage{amssymb, amscd, extarrows}
\usepackage{multirow} 

\theoremstyle{plain}\newtheorem{theorem}{Theorem}
\theoremstyle{plain}\newtheorem{definition}[theorem]{Definition}
\theoremstyle{plain}\newtheorem{example}[theorem]{Example}
\theoremstyle{plain}\newtheorem{lemma}[theorem]{Lemma}
\theoremstyle{plain}
\theoremstyle{plain}\newtheorem{proposition}[theorem]{Proposition}

\usepackage{color}

\newcommand{\nesearrow}{\nearrow\!\!\!\!\!\!\!\!\!\searrow}
\newcommand{\trs}[1]{\begin{aligned}#1\end{aligned}}

\newcommand{\tb}{\mbox{$\mathsf{tb}$}}

\newcommand{\tup}[1]{\left\langle{#1}\right\rangle}

\newcommand{\join}{\downarrow}
\newcommand{\ra}{\rightarrow}
\newcommand{\ras}{\rightarrow^\ast}
\newcommand{\rap}{\rightarrow^+} 

\newcommand{\la}  {\leftarrow}
\newcommand{\las} {\la^*}

\newcommand{\La}  {\Leftarrow}

\newcommand{\poss}{{\mathcal P}\hspace{-0.1em}os}
\newcommand{\vars}{{\mathcal V}\hspace{-0.1em}ar}

\newcommand{\F}{{\mathcal F}}
\newcommand{\V}{{\mathcal V}}
\newcommand{\T}{{\mathcal T}}
\newcommand{\R}{{\mathcal R}}
\newcommand{\U}{{\mathbb U}}

\newcommand{\usim}{{\U_{sim}}}
\newcommand{\useq}{{\U_{seq}}}
\newcommand{\uopt}{{\U_{opt}}}

\newcommand{\ie}{i.e.~}
\newcommand{\eg}{e.g.~}

\newcommand{\wrt}{w.r.t.~}

\newcommand{\set}[1]{\left\{#1\right\}}

\begin{document}

\title{Confluence of Conditional Term Rewrite Systems via Transformations}  

\author{Karl Gmeiner
\institute{Department of Computer Science, UAS Technikum Wien, Austria}
\email{gmeiner@technikum-wien.at}}

\def\titlerunning{Confluence of CTRSs via Transformations}

\def\authorrunning{K.\ Gmeiner}

\maketitle

\begin{abstract}
Conditional term rewriting is an intuitive yet complex extension of
term rewriting. In order to benefit from the simpler framework of
unconditional rewriting,
transformations have been defined to eliminate the conditions
of conditional term rewrite systems. 

Recent results provide confluence criteria
for conditional term rewrite systems via transformations, yet they are restricted
to CTRSs with certain syntactic
properties like weak left-linearity. 
These syntactic properties imply that
the transformations are sound for the given
CTRS.

This paper shows how to use transformations to prove confluence of operationally 
terminating, right-stable deterministic conditional
term rewrite systems without
the necessity of soundness restrictions. For this purpose,
it is shown that certain rewrite strategies, in particular almost
U-eagerness and innermost rewriting, always imply soundness.
\end{abstract}

\section{Introduction}


\subsection{Background and Motivation}

Conditional term rewrite systems (CTRSs) are term rewrite systems in which
rewrite rules may be bound to certain conditions. Such systems are a widely accepted
extension of unconditional term rewrite 
systems (TRSs) that has been investigated for decades but they are 
more complex than unconditional TRSs. 
Several properties of unconditional rewriting are not satisfied anymore or change their
intuitive meaning
and many criteria for TRSs cannot be applied.
Thus, there have been efforts to develop transformations
that map CTRSs into unconditional TRSs, for instance in 
\cite{jcss86-bergstra-klop, ctrs87-lncs88-giovanetti-moiso, ctrs94-lncs95-hintermeier, jsc99-viry, ppdp03-antoy-brassel-hanus}.

Transformations are supposed to simplify the original CTRS by eliminating the conditions. This way, 
properties of the CTRS can be proved by using the simpler, unconditional TRS.
Yet, for this purpose one must ensure that the rewrite relation
of the transformed TRS does not give rise to rewrite sequences that
are not possible in the original CTRS, a property called soundness.

The aim of this paper is to prove that if the transformed TRS is 
confluent, then the CTRS is also confluent,
without the necessity to also 
prove soundness. This main result is applicable
to right-stable
deterministic CTRSs that are transformed 
into terminating TRSs and it
significantly improves other, similar confluence results like the ones in 
\cite{iwc2013} and \cite{iwc2014} because there are no syntactic restrictions required
that imply soundness (in particular weak left-linearity). In fact, it also holds
for CTRSs for which the used transformation is
unsound. This result leads to a new method to prove confluence of 
CTRSs that can be easily automated and it leaves space for further improvements.


\subsection{Overview and Outline}

In order to prove properties of CTRSs using transformations, one must prove
that the transformation is suitable for the given purpose. 
\cite{alp96-marchiori} introduces the notions of \emph{soundness} and \emph{completeness}
of a certain class of transformations, so-called unravelings.
Informally, soundness means that if the transformed TRS gives rise to a rewrite 
sequence in the transformed TRS then this rewrite sequence is also possible in the 
original CTRS. Completeness is the opposite of soundness, \ie that a rewrite sequence
in the CTRS also exists in the transformed TRS.


Completeness is usually implied by the structure of transformations but 
soundness is more difficult to prove and not satisfied in general.
Yet, soundness is needed to prove properties like non-termination or confluence. 
In many papers it is proved that certain syntactic properties like (weak) left-linearity imply soundness for
a certain transformation (see \eg \cite{alp96-marchiori}). 


Soundness and confluence of the transformed system
implies confluence of the original CTRS (see \eg \cite{iwc2013}), yet
there is not yet a positive or a negative result whether soundness is
essential (although confluence of the transformed CTRS does not imply soundness 
which was shown in \cite{rta12-gmeiner-et-al}). This paper will answer this question
by first showing that innermost derivations are always
sound and then show that this in fact implies confluence if
the transformed TRS is terminating.

The following section recalls some basics
and notions of (conditional) term rewriting.
Section~\ref{sec_unravel} introduces the 
most common unravelings of CTRSs. 
In Section~\ref{sec_unsoundness} 
a rewrite strategy called \emph{almost U-eager derivations} is introduced and it is
proved that it implies soundness. Based on this, further results are shown,
in particular soundness of innermost rewrite sequences.
These results are used in 
Section~\ref{sec_conf} to
prove confluence of CTRSs. Finally, the results are
summarized and similar results in the literature and possible perspectives are discussed.

\section{Preliminaries}\label{sec_prelim}

This paper follows basic 
notions and notations as they are defined in \cite{BaNi98} and \cite{book02-ohlebusch}. Basic knowledge of (conditional) term rewriting is assumed. Some less common 
notions are recalled in the following.

The set of all terms over a signature $\F$ and an infinite but countable set
of variables $\V$ is denoted as $\T(\F, \V)$. In the following $\T$ is used
if $\F$ and $\V$ are clear from context. The set of variables
in a term $s$ is $\vars(s)$. For a set of variables $X$, 
$\vec{X}$ denotes
the sequence of variables in $X$ in some arbitrary but fixed order.
The set of positions of a term $s$ is denoted as $\poss(s)$,
$s|_p$ is the subterm of $s$ at position $p$ and
$s[t]_p$ represents the term $s$ after inserting the term $t$ at position $p$.
If $p \le q$ ($p < q$), then $q$ is below (strictly below) $p$.
Otherwise $q$ is above $p$ ($p \ge q$) or parallel to $p$ ($p \parallel q$). 

A substitution $\sigma$ is a mapping from
variables to terms that is implicitly extended to terms. 
In the following, the common
postfix notation $s\sigma$ is used for the term $s$ with the substitution $\sigma$
applied. This notation is extended to substitutions, \ie 
$\sigma \tau$ corresponds to $\sigma \tau (x) = \tau(\sigma(x))$.

A rewrite rule $\alpha$ is a pair of two terms $(l, r)$, 
denoted as $l \ra r$, where $\vars(r) \subseteq \vars(l)$. 
A term rewrite system (TRS) is a pair $\R = (\F, R)$ of a signature and a set of
rules. In the following, the signature will often be left implicit
and slightly abusing notation $\R$ will be used instead of $R$. 

A rewrite
step from a term $s$ to a term $t$ at a position $p$ using a rule $\alpha$
is denoted as $s \ra_{p, \alpha, \R} t$. Some
labels are skipped if they are clear from context or irrelevant.
A single rewrite step is written as $\ra$, the transitive closure is $\rap$,
the reflexive and transitive closure is $\ras$. $\la$ ($\las$) is the 
inverse of $\ra$ ($\ras$) and $\leftrightarrow$ ($\leftrightarrow^*$) is 
$\la \cup \ra$ ($(\la \cup \ra)^*$). A rewrite sequence $u \ras_\R v$ in some TRS $\R$ 
is normalizing if $v$ is a normal form in $\R$.

The set of \emph{one-step descendants} $q\backslash A$ of a position $q$ in a term $s$
w.r.t.~the rewrite step $A: s \ra_{p, l \ra r} t$ is the set of positions
\[
	q \backslash A = \begin{cases}
	\set{q} & \text{if } q \le p \text{ or } p \parallel q \\
	\set{p.q'.q'' \mid r|_{q'} = l|_{p'}} & \text{if } l|_{p'} \text{ is a variable and } q = p.p'.q'' \\ 
	\emptyset & \text{otherwise}
	\end{cases}
\] 
The one-step descendant relation is defined as $\set{(p, q) \mid p \in q \backslash A}$.
The descendant relation is the reflexive, transitive closure of the one-step descendant 
relation, extended to rewrite sequences. The ancestor relation is the inverse of the descendant relation.
By slight abuse of terminology a term $t|_{q'}$ will be referred to as the (one-step)
descendant of a term $s|_q$ if $q'$ is a (one-step) descendant of $q$.\footnote{
From this definition it follows that $t|_p$ is a 
one-step descendant of $s|_p$ in a rewrite step $s \ra_{p} t$. 
This case is sometimes excluded from the descendant relation.}

A conditional rule is a triple $(l, r, c)$, usually denoted
as $l \ra r \La c$ where $l, r$ are terms and $c$ is a conjunction
of equations $s_1 = t_1, \ldots, s_k = t_k$. In this paper
we only consider oriented conditional rules in which equality is defined as reducibility $\ras$.
A conditional term rewrite system (CTRS) $\R$ over some signature $\F$ consists of
conditional rules. The underlying TRS $\R_u$
contains the unconditional part of the conditional rules 
$\R_u = \set{ l \ra r \mid l \ra r \La c \in \R}$.

Let $\R_n$ be the following TRSs:
\[
	\begin{aligned}
		\R_0 &= \emptyset \\
		\R_{n + 1} &= 
		\set{ l\sigma \ra r\sigma \mid 
			l \ra r \La c \in \R \text{ and }
		s\sigma \ras_{\R_n} t\sigma \text{ for all } s \ras t \in c }
	\end{aligned}
\]
A CTRS $\R$ gives rise to the rewrite step $u \ra_\R v$ if there is an $n$
such that $u \ra_{\R_n} v$. The minimal such $n$ is the \emph{depth} 
of the rewrite step.

A conditional rule is of type 1 if
there are no extra variables ($\vars(r) \cup \vars(c) \subseteq \vars(l)$).
It is of type 3 if all extra variables occur in the conditions 
($\vars(r) \subseteq \vars(c) \cup \vars(l)$). 
A \emph{normal} conditional rule is an oriented 1-rule in which 
for every condition $s_i \ras t_i$ ($i \in \set{1, \ldots, k}$), 
$t_i$ is a ground normal form \wrt $\R_u$.
A \emph{deterministic} conditional rule is an oriented 3-rule 
$l \ra r \La s_1 \ras t_1, \ldots, s_k \ras t_k$ such that 
$\vars(s_i) \subseteq \vars(l, t_1, \ldots, t_{i-1})$ 
for all $i \in \set{1, \ldots, k}$. 
A CTRS is a deterministic CTRS (DCTRS) if all rules are 
deterministic conditional rules.

A CTRS is \emph{right-stable} if for all conditional rules 
$l \ra r \La s_1 \ras t_1, \ldots, s_k \ras t_k$,
$t_i$ is either a linear constructor term or
a ground irreducible term (w.r.t.~$\R_u$), and
$\vars(t_i) \cap \vars(l, s_1, t_1, \ldots, s_{i-1}, t_{i-1}, s_i) = \emptyset$ 
for all $i \in \set{1, \ldots, k}$. In the following only right-stable DCTRSs are considered.

\section{Unravelings}\label{sec_unravel}

Unravelings are a simple class of transformations from CTRSs to TRSs that was
introduced in \cite{alp96-marchiori}. In the same paper Marchiori also
introduces multiple specific unravelings, in particular the simultaneous 
unraveling 
$\usim$ for normal 1-CTRSs. 
This unraveling splits a conditional rule 
$\alpha: l \ra r \La s_1 \ras t_1, \ldots, s_k \ras t_k$ 
into two unconditional rules: 
\[
	\begin{aligned}
	l &\ra U^\alpha(s_1, \ldots, s_k, \overrightarrow{\vars(l)}) \\
	U^\alpha(t_1, \ldots, t_k, \overrightarrow{\vars(l)}) &\ra r
	\end{aligned}
\]

The sequential unraveling that was introduced in \cite{flops99-ohlebusch}
(a similar unraveling was already defined in \cite{techrep97-marchiori})
extends this approach to DCTRSs.

\begin{definition}[sequential unraveling $\useq$ \cite{flops99-ohlebusch}]
Given a deterministic conditional rule 
$\alpha: l \ra r \La s_1 \ras t_1, \ldots, s_k \ras t_k$, $\useq$ translates
the rule into a set of unconditional rules:
\[
	\useq(\alpha) = \set{
		\begin{aligned}
		l &\ra U^\alpha_1(s_1, \vec{X_1}) && \text{(introduction rule)}\\
		U^\alpha_1(t_1, \vec{X_1}) &\ra U^\alpha_2(s_2, \vec{X_2}) && \text{(switch rule)}\\[-0.5em]
		\vdots \qquad & \qquad \vdots && \qquad \vdots \\[-0.5em]
		U^\alpha_{k-1}(t_{k-1}, \vec{X_{k-1}}) &\ra U^\alpha_k(s_k, \vec{X_k}) && \text{(switch rule)}\\
		U^\alpha_k(t_k, \vec{X_k}) &\ra r && \text{(elimination rule)}\\
		\end{aligned}
	}
\]
where $X_i = \vars(l, t_1, \ldots, t_{i-1})$. For an unconditional rule $\alpha$, $\useq(\alpha) = \set{\alpha}$.
The unraveled CTRS $\useq(\R)$ then is defined as $\bigcup_{\alpha \in \R} \useq(\alpha)$.
\end{definition}

In the following, $\useq(\F)$ denotes the signature of the unraveled
TRS $\useq(\R)$. The new function symbols $\useq(\F) \setminus \F$ are
\emph{U-symbols}. Terms rooted by a U-symbol are \emph{U-terms}.
A terms $s$ is a \emph{mixed term} if it contains U-terms ($s \in \T(\useq(\F), \V)$, 
short $\useq(\T)$), otherwise it is an \emph{original term} ($s \in \T$). 
In U-terms of some $\useq(\R)$, the first argument
encodes the \emph{conditional argument} while the other
\emph{variable arguments} contain the variable bindings.

A rewrite step in the transformed TRS in which an introduction
(switch/elimination) rule is applied is an \emph{introduction step}
(\emph{switch step}/\emph{elimination step}).

According to the original definition an unraveling $\U$ is complete, \ie $u \ras_\R v$ implies $u \ras_{\U(\R)} v$. It is sound 
if $u \ras_{\U(\R)} v$ implies $u \ras_\R v$
for all $u, v \in \T$.

The unraveling $\useq$ encodes all variable bindings in its U-terms even if they are 
not used anymore. In \cite{pepm04-duran-lucas-meseguer-marche-urbain} the variable bindings are optimized, leading to the optimized sequential unraveling $\uopt$ 
(\cite{rta05-nishida-sakai-sakabe}). In this unraveling variables are not 
encoded if they are not required in a later condition or
the right-hand side of the conditional rule: 
\[
	\uopt(\alpha) = \set{l \ra U^\alpha_1(s_1, \vec{X_1}), 
	U^\alpha_1(t_1, \vec{X_1}) \ra U^\alpha_2(s_2, \vec{X_2}), \ldots, 
	U^\alpha_k(t_k, \vec{X_k}) \ra r} \]
where 
$X_i = \vars(l, t_1, \ldots, t_{i-1}) \cap \vars(t_{i+1}, s_{i+2} \ldots, s_k, t_k, r)$.

Optimizing the variable bindings in unravelings has advantages
in some cases because less terms have to be considered in proofs. 
In \cite{gmeiner-phd-thesis} 
several soundness results for $\uopt$ are shown, in particular
soundness for U-eager rewrite sequences. 
Formally, a derivation $u_0 \ra_{p_0} u_1 \ra_{p_1} \cdots \ra_{p_{n-1}} u_n$
in some $\U(\R)$ 
is \emph{U-eager} if U-terms are immediately rewritten, i.e., 
$p \le p_i$ for all U-terms $u_i|_p$.

Yet, this optimization has some drawbacks. For instance, two terms
that are not joinable in the original CTRS rewrite to the same mixed term because
a variable binding is erased. Because of this phenomenon, the main result of this paper
does not hold for $\uopt$.

\begin{example}[unsoundness for confluence of the optimized unraveling]\label{ex_unsound_uopt}
	Consider the following DCTRS and its transformed terminating TRS using the optimized
	unraveling:
	\[
	\R = \set{
		\trs{
			a & \ra s(b) \\[-0.5em]
			& \searrow \\[-0.5em]
			& \mathrel{\quad} s(c) \\
			s(x) &\ra A \La B \ras C
		}} \qquad
		\uopt(\R) = \set{
		\trs{
			a & \ra s(b) \\[-0.5em]
			& \searrow \\[-0.5em]
			& \mathrel{\quad} s(c) \\
			s(x) &\ra U^\alpha_1(B) \\
			U^\alpha_1(C) &\ra A
		}}
	\]
	
	$\uopt(\R)$ is confluent since the only critical pair $\tup{s(b), s(c)}$
	gives rise to the common reduct $U^\alpha_1(B)$ and the transformed TRS
	is terminating. Yet,
	$\R$ is not confluent because $a$ rewrites to $s(b)$ and $s(c)$ but the condition
	of the conditional rule is never satisfied so that $s(b)$ and $s(c)$ 
	are irreducible.
\end{example}

Since $\useq$ preserves all variable bindings of the left-hand side of a conditional rule
it is possible to extract these bindings and insert them into the corresponding left-hand side, thus
obtaining the back-translation $\tb$ (defined in \cite{rta10-gmeiner-et-al}, similar mappings 
are used in the proofs in \cite{alp96-marchiori} and \cite{book02-ohlebusch}).

\begin{definition}[back-translation $\tb$]
	Let $\R = (\F, R)$ be a CTRS, then $\tb: \useq(\T) \mapsto \T$ is 
	defined as follows:
	\[
	\tb(s) = \begin{cases}
	s & \text{if } s \text{ is a variable} \\
	f(\tb(s_1), \ldots, \tb(s_k)) & \text{if } s = f(s_1, \ldots, s_k) 
	\text{ and } f \in \F \\
	l\sigma & \text{if } s = U^\alpha_i(w, v_1, \ldots, v_m) \text{ and}\\
	& \alpha = l \ra r \La s_1 \ras t_1, \ldots s_k \ras t_k 
	\end{cases}
	\]
	where $\sigma$ is defined as $x_i\sigma = \tb(v_i)$
	where $\overrightarrow{\vars(l, t_1, \ldots, t_{i-1})} = 
	x_1, \ldots, x_m$.
\end{definition}

In the following $\tb$ will sometimes
be extended to substitutions ($x\,\tb(\sigma) = \tb(x\sigma)$ for $x \in \mathcal{D}om(\sigma)$).
The back translation allows us to define soundness such that it also extends to mixed terms: 
A rewrite sequence $u \ras_\useq(R) v$ ($u \in \T$) is sound if $u \ras_\R \tb(v)$.

\section{Soundness and Completeness of Transformations} \label{sec_unsoundness}

The transformation $\useq$ is not sound for DCTRSs in general. This was first shown
by Marchiori in \cite{alp96-marchiori} using a normal 1-CTRS that consists of multiple non-linear rules. 
For DCTRSs we presented another example in \cite{rta12-gmeiner-et-al}.

\begin{example}[unsoundness \cite{rta12-gmeiner-et-al}]\label{ex_unsound}
Consider the following DCTRS and its unraveling 
\[
	\R = \set{
	\trs{
        a &\to c    \\[-0.5em]
        &\nesearrow     \\[-0.5em]
        b &\to d  \\
        s(c)& \ra t(k) \\[-0.5em]
        & \searrow \\[-0.5em]
        & \mathrel{\quad} t(l) \\
        g(x,x) &\ra h(x,x) \\
		f(x) &\to \tup{x, y} \Leftarrow s(x) \ras t(y)
	}} \qquad
		\useq(\R) = \set{
	\trs{
        a &\to c    \\[-0.5em]
        &\nesearrow     \\[-0.5em]
        b &\to d  \\ 
        s(c)& \ra t(k) \\[-0.5em]
        & \searrow \\[-0.5em]
        & \mathrel{\quad} t(l) \\
		g(x,x) &\to h(x,x) \\
		f(x) &\to U^\alpha_1(s(x), x) \\
		U^\alpha_1(t(y), x) &\to \tup{x, y}
	}}
\]

In $\useq(\R)$, there is the following reduction sequence:
\[ g(f(a), f(b)) \ras g(U^\alpha_1(s(c), d), U^\alpha_1(s(c), d)) \ra h(U^\alpha_1(s(c), d), U^\alpha_1(s(c), d)) 
\ras h(\tup{d, k}, \tup{d, l})
\]

Yet, this derivation is not possible in $\R$ because
there is no common reduct of $f(a)$ and $f(b)$ that rewrites to both,
$\tup{d,k}$ and $\tup{d,l}$.
\end{example}

The CTRSs of Example~\ref{ex_unsound} and the counterexample of \cite{alp96-marchiori} are 
syntactically very complex. Based on this observation it was shown that many syntactic properties imply soundness:
Left-linearity (normal 1-CTRSs \cite{alp96-marchiori}/\cite{book02-ohlebusch}, DCTRSs \cite{rta11-nishida-et-al}),
weak left-linearity, right-linearity (normal 1-CTRSs \cite{rta10-gmeiner-et-al}, DCTRSs \cite{rta12-gmeiner-et-al}),
non-erasingness (normal 1-CTRSs \cite{rta10-gmeiner-et-al}, 2-DCTRSs \cite{rta12-gmeiner-et-al}, counterexample for 3-DCTRSs \cite{rta12-gmeiner-et-al})
and weak right-linearity (DCTRSs \cite{gmeiner-phd-thesis}).

The CTRS of Example~\ref{ex_unsound} is not confluent and in \cite{rta10-gmeiner-et-al} it is shown
that the simultaneous unraveling is sound for confluent normal 1-CTRSs. Yet,
this result does not hold for DCTRSs:

\begin{example}[unsoundness for confluence \cite{rta12-gmeiner-et-al}]\label{ex_unsound_conf}
Let $\R$ be the CTRS of Example~\ref{ex_unsound} and $\R'$ be the CTRS
consisting of the unconditional rules
\[
	\R' = \set{c \ra e \la d, k \ra e \la l, s(e) \ra t(e)}
\]

Then, $\R \cup \R'$ and $\useq(\R \cup \R')$ are confluent,
yet, the argument of Example~\ref{ex_unsound} still holds
so that the reduction sequence $g(f(a),f(b)) \ras h(\tup{d, k}, \tup{d, l})$
in the transformed TRS is still unsound.
Nonetheless, the last term of the unsound derivation can be further
reduced to the irreducible term $h(\tup{e, e}, \tup{e, e})$. The 
derivation $g(f(a),f(b)) \ras h(\tup{e, e}, \tup{e, e})$ is sound.
\end{example}

Although the previous example shows that confluence of the transformed
TRS is not sufficient for soundness, it also shows that (in contrast
to Example~\ref{ex_unsound}) the last term of the unsound derivation
can be further reduced. In fact, all normalizing derivations in confluent DCTRSs 
are sound \cite{rta12-gmeiner-et-al}.


The original definition of unravelings in \cite{alp96-marchiori}
states that an unraveling must be complete and preserve the 
original signature. Based on the definition of the unravelings
it is not surprising that completeness is satisfied in all cases.
In the following the proof of \cite{alp96-marchiori} is adapted to 
$\useq$. The proof will be useful to
motivate a rewrite strategy that implies soundness.

\begin{lemma}[completeness of $\useq$]\label{lm_complete}
	Let $\R$ be an oriented CTRS and $s, t$ be two
	terms such that $s \ra_\R t$, then $s \ra^+_{\useq(\R)} t$.
\end{lemma}

\begin{proof}
	By induction on the depth $n$ of the rewrite step $s \ra_{n, \R} t$. 
	If  $n = 0$, then the applied rule $\alpha$ is an unconditional
	rule and $\alpha \in \useq(\R)$.
	
	Otherwise, let 
	$\alpha: l \ra r \La s_1 \ras t_1, \ldots, s_k \ras t_k$ be the rule applied in
	$s \ra_{n, \R} t$ 
	so that	$s = C[l\sigma]$ and $t = C[r\sigma]$. By the definition of the depth, 
	$s_i\sigma \ras_{\R_{n - 1}} t_i\sigma$
	for all $i \in \set{1, \ldots, k}$. By the induction hypothesis, there are
	derivations $s_i\sigma \ras_{\useq(\R)} t_i\sigma$. Thus, there is the following derivation
	in $\useq(\R)$:
	\[
		\begin{aligned}
		l\sigma &\ra U^\alpha_1(s_1 \sigma, \vec{X}_1\sigma) \ras U^\alpha_1(t_1 \sigma, \vec{X}_1\sigma)
		\ra U^\alpha_2(s_2 \sigma, \vec{X}_2\sigma) \ras U^\alpha_2(t_2 \sigma, \vec{X}_2\sigma) \ra \cdots \\
		&\ra 
		 U^\alpha_1(s_k \sigma, \vec{X}_k\sigma) \ras 
		U^\alpha_k(t_k \sigma, \vec{X}_k\sigma) \ra r\sigma
		\end{aligned}
	\]
\end{proof}


The previous completeness result constructs a derivation in $\useq(\R)$ in which
first the U-term is introduced, then the conditional argument is rewritten and
finally the U-term is eliminated. The definition of the U-eager rewrite strategy
is based on such derivations but it also allows rewrite steps inside variable bindings. 

In U-eager derivations, after a U-term is introduced only rewrite steps
inside this U-term are allowed until it is eliminated. 
Rewrite steps outside of U-terms are forbidden. The reason for this limitation is that in a derivation in some $\uopt(\R)$ 
one obtains mixed terms that have no meaning in the original CTRS. For instance, in Example~\ref{ex_unsound_uopt}
the mixed term $U^\alpha_1(B)$ is a common reduct of $s(b)$ and $s(c)$. Yet, there is no such term in the 
original CTRS.

For $\useq$, mixed terms can be back-translated to the left-hand
side of the conditional rule because all variable bindings are preserved. Therefore, U-eagerness for some $\useq(\R)$ can be 
generalized to also allow rewrite steps outside of U-terms even if they are not eliminated. In 
such \emph{almost U-eager rewrite sequences} if a U-term is not rewritten it is considered to 
represent a failed conditional evaluation and thus the arguments
of such a U-term and the U-term itself must not be rewritten anymore. Rewrite steps
above such U-terms, including erasing rewrite steps, are allowed. 

\begin{definition}[almost U-eager derivations]\label{df_almost_ueager}
Let $\R$ be a DCTRS. 
A derivation $D: u_0 \ra_{p_0} u_1 \ra_{p_1} \cdots \ra_{p_{n-1}} u_n$ in 
$\useq(\R)$ is \emph{almost U-eager}, if for every rewrite
step $u_i \ra_{p_i} u_{i+1}$, if there is a U-term $u_i|_q$ such that
$q \le p_i$, then also $q \le p_{i-1}$ ($i \in \set{1, \ldots, n-1}$).
\end{definition}

This way, rewrite steps in U-terms are always grouped in such derivations which makes tracking terms easier.
Furthermore, U-terms that represent intermediate evaluation steps of conditions are 
isolated from other rewrite steps. Rewrite steps above U-terms that are rewritten in a later rewrite step 
are vitally important for
unsoundness. Observe that the unsound
derivation of Example~\ref{ex_unsound} is not almost U-eager and that in the unsound derivation a non-linear
rewrite step is applied above a U-term.


The proof for soundness of almost U-eager derivations will use
the same proof structure that was already used in \cite{rta12-gmeiner-et-al}. 
First, we recall the following lemma
that states that rewrite steps in variable and conditional
arguments can be extracted from 
derivations. 

\begin{lemma}[extraction lemma of \cite{rta12-gmeiner-et-al}]\label{lm_extract}
Let $\R$ be a DCTRS and $D: u_0 \ra_{p_0} u_1 \ra_{p_1} \cdots \ra_{p_{n-1}} u_n$ be a derivation
in $\useq(\R)$ ($u_0 \in \T$). If $u_n|_{p} = U^\alpha_i(w, \vec{X_i}\sigma_{i+1})$
where $\alpha$ is the conditional rule $l \to r \La s_1 \ras t_1, \ldots, s_k \ras t_k$,
then there is an index $m$ and a position $q$ such that $u_m|_q$ is an ancestor of 
$u_n|_p$ and there are 
substitutions $\sigma_1, \ldots, \sigma_i$ such that $u_m|_q = l\sigma_1$ and 
the following derivations can be extracted from $D$:
\begin{itemize}
	\item $s_j\sigma_j \ras_{\useq(\R)} t_j\sigma_{j+1}$ ($j \in \set{1, \ldots, i - 1}$),
	\item $x\sigma_j \ras_{\useq(\R)} x\sigma_{j+1}$ ($j \in \set{1, \ldots, i}$, $x \in X_j$), and
	\item $s_i\sigma_i \ras_{\useq(\R)} w$.
\end{itemize}

Furthermore, in the reductions above for every single rewrite step $u \ra v$ 
there is an index
$m' \in \set{m + 1, \ldots, n - 1}$ and a position $q'$ such that $u_{m'}|_{q'} = u$ and $u_{m'+1}|_{q'} = v$.
\end{lemma}

In the following this extraction lemma will be used implicitly.

Next, a monotony result on $\tb$ is shown.

\begin{lemma}[monotony of $\tb$]\label{lm_tb_mon}
Let $\R$ be a DCTRS. If $u \ra_{p,\useq(\R)} v$
for $u, v \in \useq(\T)$ and $\tb(u|_p) \ras_\R \tb(v|_p)$
then $\tb(u|_q) \ras_\R \tb(v|_{q'})$ for all $q \in \poss(u)$
and descendants $q'$ of $q$.
\end{lemma}

\begin{proof}
By case distinction on $p$ and $q$: If $p < q$ or $p \parallel q$,
then $u|_q = v|_{q'}$, hence also $\tb(u|_q) = \tb(v|_{q'})$.

Otherwise, if $q \le p$, then there is only one descendant of $u|_q$ which is $v|_q$.
Let $q.q' = p$. Then by induction on $|q'|$,
if $q = p$ then $\tb(u|_q) \ras_\R \tb(v|_{q'})$ is equivalent
to the assumption $\tb(u|_p) \ras_\R \tb(v|_{p})$.

For the induction step, let $q' = i.q''$. There are the following cases based on
the term $u|_q$: If $u|_q = f(u_1, \ldots, u_n)$ where $f \in \F$ is an original 
symbol, then $\tb(u|_q) = f(\tb(u_1), \ldots, \tb(u_n))$ and
$\tb(v|_q) = f(\tb(u_1), \tb(u_{i-1}), \tb(v|_{q.i}), \tb(u_{i+1}), \ldots, \tb(u_n))$.
By the induction hypothesis $\tb(u_i) \ras \tb(v|_{q.i})$, thus also 
$\tb(u|_q) \ras \tb(v|_q)$.

The remaining case is that $u|_q$ is a U-term $U^\alpha_j(w, x_1, \ldots x_n)\sigma$.
If $i = 1$, then the rewrite step is inside the conditional argument so that the variable 
bindings are unmodified and $\tb(u|_q) = \tb(v|_q)$.
Otherwise, $v|_q = U^\alpha_j(w, x_1, \ldots x_n)\sigma'$ where $x_j\sigma = x_j\sigma'$
for all $j \in \set{1, \ldots, i - 2, i, \ldots, n}$. By the induction hypothesis, 
$\tb(x_{i-1}\sigma) \ras \tb(x_{x_{i-1}}\sigma')$. Hence,
$\tb(u|_q) = l\,\tb(\sigma)$ where $\tb(v|_q) = l\,\tb(\sigma')$ and thus
$\tb(u|_q) \ras \tb(v|_q)$.
\end{proof}

In the next lemma, single rewrite steps of a derivation are translated using $\tb$.

\begin{lemma}[technical key lemma]\label{lm_almost_ueager_sound}
	Let $\R$ be a right-stable DCTRS and
	let $u_0 \ra_{p_0} u_1 \ra_{p_1} \cdots \ra_{p_{n-1}} u_n$ be an almost
	U-eager derivation in $\useq(\R)$ where $u_0 \in \T$. Then, 
	$\tb(u_i|_{p_i}) \ras_\R \tb(u_{i+1}|_{p_i})$ ($i \in \set{0, \ldots, n-1}$).
\end{lemma}

\begin{proof}
In the following, assume w.l.o.g.~that for all substitutions, mapped terms
do not share variables with the domain, \ie,~$\mathcal{D}om(\sigma) 
\cap \vars(x\sigma) = \emptyset$ for all $x \in \mathcal{D}om(\sigma)$.

By induction on the length of the derivation $n$: If $n = 0$, the result holds
vacuously.

Otherwise $\tb(u_i|_q) \ras_\R \tb(u_{i+1}|_{q'})$ for all 
one-step descendants $u_{i+1}|_{q'}$ of
$u_i|_q$ by the induction hypothesis and Lemma~\ref{lm_tb_mon}. 
Consequently also $\tb(u_i|_q) \ras_\R \tb(u_j|_{q''})$ for all
descendants $u_j|_{q''}$ of $u_i|_q$  ($1 \le i < j < n$).
 
By case distinction on the rule applied in the 
last rewrite step $u_{n-1} \ra_{\alpha, p_{n-1}} u_{n}$:
If the applied rule is an unconditional original rule $l \to r \in \R$, then 
$u_{n-1}|_{p_{n-1}} = l\sigma$,
$u_{n}|_{p_{n-1}} = r\sigma$,
$\tb(u_{n-1}|_{p_{n-1}}) = l\,\tb(\sigma)$ and 
$\tb(u_{n}|_{p_{n-1}}) = r\,\tb(\sigma)$.

If the applied rule is an introduction rule or a switch rule, 
$\tb(u_{n-1}|_{p_{n-1}}) = \tb(u_{n}|_{p_{n-1}})$.

Finally, if the applied rule is an elimination rule, then by the definition
of almost U-eagerness, all preceding rewrite steps are below $p_{n-1}$ up to
the introduction step of the U-term, i.e., if the conditional rule
is $\alpha: l \ra r \La s_1 \ras t_1, \ldots, s_k \ras t_k$, then 
there is an $m$ such that $u_m|_{p_m} = l\sigma_1$, $p_m = p_{n-1}$, $p_m \le p_{i}$
for all $i \in \set{m, \ldots, n-1}$ and the derivation $u_m|_{p_m} \ras_{\useq(\R)} u_{n}|_{p_{n-1}}$
is
\[
\begin{aligned}
l\sigma_1 \ra U^\alpha_1(s_1\sigma_1, \vec{X_1}\sigma_1) 
\ras U^\alpha_1(t_1\sigma_2, \vec{X_1}\sigma_2) &\ra U^\alpha_2(s_2\sigma_2, \vec{X_2}\sigma_2) \ras \cdots \\
\ras U^\alpha_k(t_{k}\sigma_{k+1}, \vec{X_k}\sigma_{k+1}) &\ra r\sigma_{k+1}
\end{aligned}\]

By the induction hypothesis, $\tb(x\sigma_i) \ras_\R \tb(x\sigma_{i+1})$ and
$\tb(s_i\sigma_i) \ras_\R \tb(t_i\sigma_{i+1})$ for all $x \in X_i$.

Let $\sigma$ be the combined substitution $\tb(\sigma_{1}) \tb(\sigma_2) \cdots \tb(\sigma_{k+1})$,
then $s_i \sigma \ras_\R s_i\,\tb(\sigma_i)$ and $t_i\,\tb(\sigma_{i+1}) = t_i \sigma$ by
right-stability. Hence, the conditions are satisfied for $\sigma$ and $l\sigma \ra_\R r\sigma$. Furthermore, $l\sigma = l\sigma_1$.
Thus, $l\,\tb(\sigma_1) \ra_\R r\sigma \ras_\R r\,\tb(\sigma_{k+1})$.
\end{proof}

Finally we prove soundness of almost U-eager 
rewrite sequences.

\begin{lemma}[soundness of almost U-eager derivations]\label{lm_sound_eager}
	Let $\R$ be a right-stable DCTRS. If
	$u_0 \ra_{p_0} u_1 \cdots \ra_{p_{n-1}} u_n$ is an almost
	U-eager derivation in $\useq(\R)$ ($u_0 \in \T$) then
	$u_0 \ras_\R \tb(u_n)$.
\end{lemma}

\begin{proof}
By induction on the length of the derivation, if $n = 0$ the result holds vacuously.
Otherwise, by Lemma~\ref{lm_almost_ueager_sound}, $\tb(u_{n-1}|_{p_{n-1}}) \ras_\R \tb(u_n|_{p_{n-1}})$.
By Lemma~\ref{lm_tb_mon}, $\tb(u_{n-1}) \ras_\R \tb(u_n)$. Since by the inductive hypothesis
$u_0 \ras_\R \tb(u_{n-1})$ finally $u_0 \ras_\R \tb(u_n)$.
\end{proof}

This result can
be used to prove soundness for other rewrite strategies. Next,
it is shown that innermost derivations can be converted 
into almost U-eager derivations,
thus proving soundness of innermost rewriting. For this reason, innermost derivations
are translated into almost U-eager derivations.



\begin{lemma}[innermost to almost-U-eager]\label{lm_innermost_ueager}
Let $\R$ be a DCTRS and let 
$u_0 \ra_{p_0} u_1 \ra_{p_1} \cdots \ra_{p_{n-1}} u_n$ be an innermost derivation ($u_0 \in \T$).
Then there is
an innermost, almost U-eager derivation $u_0 \ras_{\useq(\R)} u_n$.
\end{lemma}

\begin{proof}
By induction on the length $n$ of the derivation.
If $n = 0$ the result holds vacuously.
Otherwise, by the induction hypothesis, 
the derivation $u_0 \ras_{\useq(\R)} u_{n-2} \ra_{p_{n-2}, \useq(\R)} u_{n-1}$ is 
innermost and almost U-eager.

By case distinction on the last rewrite step $u_{n-1} \ra_{p_{n-1}} u_n$:
If $p_{n-1}$ is not below a U-term, then $u_0 \ras_{\useq(\R)} u_n$ is already almost U-eager.
Otherwise, there is a U-term $u_{n-1}|_q$ and $q \le p_{n-1}$. 
If there are multiple nested U-terms, let $u_{n-1}|_q$ be the innermost such U-term.
By case distinction on $p_{n-2}$ and $q$: The case $p_{n-2} < q$ is not possible
because of the assumption that the derivation is innermost. 

If $q \le p_{n-2}$, then $u_0 \ras_{\useq(\R)} u_n$ is almost U-eager.

If $q \parallel p_{n-2}$, then let $m$ be the largest value such that 
$p_{n-m}, p_{n-m+1}, \ldots, p_{n-2}$ are parallel to $q$. Since
$u_0 \in \T$, $m < n$. Then, $u_{n-m}|_{q} = u_{n-1}|_q$ and 
$q \le p_{n-m-1}$.  Therefore, the following rewrite sequence in 
$\useq(\R)$ is in fact U-eager:
\[
\begin{aligned}
u_0 \ras u_{n-m-1} &\ra_{p_{n-m-1}} u_{n-m} \ra_{p_{n-1}} u_{n-m}[u_n|_{p_{n-1}}]_{p_{n-1}} \ra_{p_{n-m}} \\
    &\ra_{p_{n-m}} u_{n-m+1}[u_n|_{p_{n-1}}]_{p_{n-1}} \ra_{p_{n-m+1}} u_{n-m+2}[u_n|_{p_{n-1}}]_{p_{n-1}} \ra_{p_{n-m+2}} \cdots \\
    & \ra_{p_{n-3}} u_{n-2}[u_n|_{p_{n-1}}]_{p_{n-1}} \ra_{p_{n-2}} u_n 
\end{aligned}
\]
\end{proof}

Since almost U-eager rewrite sequences are sound this implies soundness.

\begin{lemma}[soundness of innermost derivations]\label{lm_sound_innermost}
Let $\R$ be a right-stable DCTRS. 
Let $u \ras_{\useq(\R)} v$ be an innermost derivation such 
that $u \in \T$. Then, $u \ras_{\R} \tb(v)$.
\end{lemma}

\begin{proof}
By Lemma~\ref{lm_innermost_ueager}, there is an almost U-eager
derivation $u \ras_{\useq(\R)} v$. By Lemma~\ref{lm_sound_eager},
$u \ras_{\R} \tb(v)$.
\end{proof}

\begin{theorem}[soundness of innermost derivations]\label{tm_innermost}
$\useq$ is sound for innermost derivations for right-stable DCTRSs.
\end{theorem}

\begin{proof}
By Lemma~\ref{lm_sound_innermost}, if $u \ras_{\useq(\R)} v$ is
an innermost derivation, then there is a derivation $u \ras_\R \tb(v)$.
\end{proof}

Innermost derivations are therefore sound. Nonetheless, innermost rewriting
is not suitable to simulate conditional rewriting in general because
they are not complete. This can be easily seen in CTRSs in which the
conditions are satisfiable but not innermost-satisfiable.

\begin{example}[incompleteness of innermost rewriting]
Consider the following CTRS and its unraveled TRS:
\[
	\R = \set{
	\trs{
        a &\to b \\
        f(a) &\to b \\
        A &\to B \Leftarrow f(a) \ras b
	}} \qquad
		\useq(\R) = \set{
	\trs{
        a &\to b \\
        f(a) &\to b \\
        A &\to U^\alpha_1(f(a)) \\
        U^\alpha_1(b) &\to B
	}}
\]

In $\R$, the condition $f(a) \ras b$ is satisfied (although there is no innermost
derivation $f(a) \ras_\R b$), therefore, $A$ rewrites to $B$. 
This derivation is innermost (yet, notice that the conditional evaluation is not). 
Nonetheless, in $\useq(\R)$ the only innermost derivation starting from $A$ is
$A \ra U^\alpha_1(f(a)) \ra U^\alpha_1(f(b))$ where the last term is irreducible.
In particular, there is no innermost derivation for $A \ras_{\useq(\R)} B$.
\end{example}

Nonetheless, we obtain completeness if the
transformed TRS is confluent and terminating:

\begin{proposition}[completeness for innermost rewriting]
Let $\R$ be a right-stable DCTRS such that $\useq(\R)$ is 
confluent and terminating. Then, 
if $u \ras_{\R} v$ ($u, v \in \T$) such
that $v$ is irreducible (w.r.t.~$\R$), then 
there is an innermost derivation $u \ras_{\useq(\R)} v'$
such that $\tb(v') = v$.
\end{proposition}

\begin{proof}
Because of completeness of $\useq$, there is a derivation
$u \ras_{\useq(\R)} v$. By confluence and termination there is 
a unique normal form $w \in \useq(\T)$ of $u$ and $v$ in $\useq(\R)$ and
there is an innermost derivation $u \ras_{\useq(\R)} w$.

Finally,
the assumption that $v$ is a normal form in $\R$ and Lemma~\ref{lm_sound_innermost} imply that
$\tb(v) = w'$ for all $w' \in \useq(\T)$ such that 
$v \ras_{\useq(\R)} w'$.
\end{proof}

Next, we prove soundness for DCTRSs that are transformed into confluent and
terminating TRSs. For this purpose, 
observe that if a
TRS is confluent and terminating, then for every
derivation $u \ras v$ such that $v$ is a normal form there is an innermost 
derivation $u \ras v$. This observation can be combined with
Theorem~\ref{tm_innermost} that states that innermost, normalizing rewrite sequences in some right-stable DCTRS are sound: 

\begin{lemma}[soundness for confluent and terminating TRSs]\label{lm_sound_norm_conf}
	Let $\R$ be a right-stable DCTRS such that $\useq(\R)$ is terminating and confluent
	and let $u \ras_{\useq(\R)} v$ be a normalizing rewrite sequence ($u \in \T$). 
	Then, $u \ras_\R \tb(v)$.
\end{lemma}

\begin{proof}
	$\useq(\R)$ is terminating and confluent, and $v$ is a normal form
	in the derivation $u \ras_{\useq(\R)} v$. Therefore, there is an innermost derivation
	$u \ras_{\useq(\R)} v$. By Lemma~\ref{lm_sound_innermost} this implies $u \ras_\R \tb(v)$.
\end{proof}

\begin{theorem}[soundness for normalizing rewrite sequences]\label{tm_sound_norm}
	Let $\R$ be a right-stable DCTRS such that $\useq(\R)$ is confluent and
	terminating. Then $\useq$ is sound for reductions to normal forms.
\end{theorem}

\begin{proof}
	Straightforward from Lemma~\ref{lm_sound_norm_conf}.
\end{proof}

The previous theorem is interesting because it shows that 
\cite[Theorem 9]{rta12-gmeiner-et-al} (soundness for reductions to normal forms
of confluent DCTRSs), also holds if only the transformed TRS is known to be confluent.

\section{Confluence of Conditional Term Rewrite Systems}\label{sec_conf}

Our goal is to prove that if $\useq(\R)$ is confluent, then also
$\R$ is confluent. For this purpose we introduce another soundness property, 
soundness for joinability.

\begin{definition}[soundness for joinability]
	An unraveling $\U$ is sound for joinability for a CTRS $\R$ if
	for all terms $u, v \in \T$ such that $u \join_{\U(\R)} v$
	also $u \join_\R v$.
\end{definition}

Soundness for joinability is important in connection with confluence
because it allows us to prove confluence of a DCTRS via confluence of the transformed
TRS.

There is an important connection between soundness for joinability and confluence.

\begin{lemma}[soundness for joinability and confluence]\label{lm_join_conf}
Let $\R$ be a CTRS such that $\useq(\R)$ is confluent
and $\useq$ is sound for joinability, then $\R$ is confluent.
\end{lemma}

\begin{proof}
Consider two terms $u, v \in \T$ such that $u \leftrightarrow^*_\R v$.
Since $\useq$ is complete by Lemma~\ref{lm_complete}, $u \leftrightarrow^*_{\useq(\R)} v$. 
$\useq(\R)$ is confluent so that $u \join_{\useq(\R)} v$.
By soundness for joinability this implies $u \join_\R v$.
\end{proof}

It remains to prove soundness for joinability of right-stable DCTRSs for which
the transformed TRS is confluent. 
Theorem~\ref{tm_sound_norm} shows that confluence and termination of the transformed
TRS imples soundness for normalizing derivations. Since every term is terminating
this implies soundness for joinability:

\begin{lemma}[soundness for joinability]\label{lm_sound_join}
Let $\R$ be a right-stable DCTRS such that $\useq(\R)$ is confluent and 
terminating, and let $u \join_{\useq(\R)} v$ ($u, v \in \T$), then 
$u \join_\R v$.
\end{lemma}

\begin{proof}
Since $\useq(\R)$ is confluent and terminating, $u \ras_{\useq(\R)} v$
implies that there is an irreducible term $w \in \useq(\T)$ such that
$u \ras_{\useq(\R)} w \las_{\useq(\R)} v$. Since $\R$ is
right-stable, Lemma~\ref{lm_sound_norm_conf} implies
$u \ras_\R \tb(w) \las_\R v$. 
\end{proof}

Thus we obtain our main result:

\begin{theorem}[soundness for confluence]\label{tm_conf}
Let $\R$ be a right-stable DCTRS such that $\useq(\R)$ is
confluent and terminating, then $\R$ is confluent.
\end{theorem}

\begin{proof}
By Lemma~\ref{lm_sound_join}, $\useq$ is sound for joinability
for $\R$. Since $\useq(\R)$ is confluent, Lemma~\ref{lm_join_conf} implies
that $\R$ is confluent.
\end{proof}

This confluence result is remarkable because it also holds for CTRSs for which 
$\useq$ is unsound like the CTRS of Example~\ref{ex_unsound_conf}. 

\begin{example}[unsound confluent CTRS]
Let us recall the right-stable DCTRS of Example~\ref{ex_unsound_conf} and its
transformed TRS.
\[
\R = \set{
	\trs{
		a &\to c \ra e    \\[-0.5em]
		&\nesearrow \;\;\;\, \nearrow     \\[-0.5em]
		b &\to d  \\
		k& \ra e \\[-0.5em]
		& \nearrow \\[-0.5em]
		l& \\
		s(c)& \ra t(k) \\[-0.5em]
		& \searrow \\[-0.5em]
		& \mathrel{\quad} t(l) \\
		s(e) &\ra t(e) \\
		g(x,x) &\ra h(x,x) \\
		f(x) &\to \tup{x, y} \Leftarrow s(x) \ras t(y)
	}} \qquad
	\useq(\R) = \set{
		\trs{
		a &\to c \ra e    \\[-0.5em]
		&\nesearrow \;\;\;\, \nearrow     \\[-0.5em]
		b &\to d  \\
		k& \ra e \\[-0.5em]
		& \nearrow \\[-0.5em]
		l& \\
		s(c)& \ra t(k) \\[-0.5em]
		& \searrow \\[-0.5em]
		& \mathrel{\quad} t(l) \\
		s(e) &\ra t(e) \\
		g(x,x) &\to h(x,x) \\
		f(x) &\to U^\alpha_1(s(x), x) \\
		U^\alpha_1(t(y), x) &\to \tup{x, y}
		}}
	\]

	The transformed TRS $\useq(\R)$ is confluent because it is 
	terminating and all critical pairs are joinable. Therefore,
	by Theorem~\ref{tm_conf}, $\R$ is also confluent.
\end{example}

Although termination of the transformed TRS seems to be a major limitation, 
\cite{iwc2013} proves that for an unraveling similar to $\useq$, 
(weakly-)left-linearity (which implies soundness) and
confluence of the transformed TRS implies confluence of the original CTRS. 
Currently it is not known whether Theorem~\ref{tm_conf}
also holds for DCTRSs that are transformed into non-terminating and non-left-linear TRSs.

\section{Conclusion}\label{sec_conc}

\subsection{Summary}

Transformations have been used as a tool to prove termination and confluence
of conditional term rewrite systems for a long time.
For confluence the problem is that 
the rewrite relation of the transformed system may give
rise to rewrite sequences that are not possible in the original system, 
\ie the transformation may not be sound. 

We use the so-called sequential unraveling, a simple transformation for
deterministic CTRSs that was introduced in \cite{flops99-ohlebusch} based
on \cite{techrep97-marchiori}.

Recent results (\eg in \cite{iwc2013}) show that confluence of the transformed
system (using the sequential unraveling) 
implies confluence of the original system if the transformation is sound.
There are many syntactic restrictions like (weak) left-linearity that
imply soundness, yet, for non-left-linear CTRSs for which
the transformation is not sound there are no such results yet.
Lemma~\ref{lm_sound_norm_conf} shows that if the transformed system
is terminating and confluent, normalizing
derivations are always sound. This 
result is interesting because a similar result was shown in \cite{rta12-gmeiner-et-al} for confluent CTRSs. 
 
This lemma holds because innermost rewrite sequences
in the transformed system are always sound (Theorem~\ref{tm_innermost}). 
Since soundness for normalizing derivations
implies soundness for joinability (which
implies soundness for confluence)
we finally can show that 
a right-stable, deterministic CTRS is confluent if the transformed TRS is
confluent and terminating  (Theorem~\ref{tm_conf}). 

It should be pointed out that it is not yet known whether termination
is really needed in this result. If there is a counterexample for this 
we know that it must be non-left-linear, non-terminating
and confluent.

\subsection{Related Work and Perspectives}

In \cite{iwc2013}, we presented a confluence criterion for CTRSs based on 
soundness and confluence of the transformed system for an unraveling
similar to $\useq$. 
\cite{iwc2014} contains a similar result for the
structure preserving transformation of \cite{rta06-serbanuta-rosu}.

Yet all these results have in common that they require some syntactic criterion
like (weakly) left-linearity of the CTRS that implies soundness. 
Theorem~\ref{tm_conf} is a significant improvement to these results because
it is also applicable to non-linear 
CTRSs for which the transformation is unsound. 

There are many confluence results for CTRSs in the literature and one similar result
is \cite[Theorem 4.1]{lpar94-avenhaus-loria-saenz}, stating that 
every strongly deterministic TRS that is quasi-reductive and has joinable critical
pairs is confluent. This result does not use
transformations but it can be seen that critical pairs in the CTRS correspond
to one or more critical pairs in the transformed system while
termination of the transformed TRS implies quasi-reductiveness 
\cite{wst07-schernhammer-gramlich}. Hence,
it subsumes Theorem~\ref{tm_conf}.

Yet, Theorem~\ref{tm_conf} has some advantages over 
\cite[Theorem 4.1]{lpar94-avenhaus-loria-saenz}. In particular,
it does not use the framework of conditional rewriting. Checking
for joinability of terms in CTRSs is easier in the transformed
unconditional TRS which is important for automated
confluence proofs.

The main result does not extend any previous results but rather
is a novel approach to prove confluence. It uses a
simple transformation and a very general proof structure.
Hence, the result might be improved in the future \eg by 
relaxing the requirements for confluence or termination. Termination 
is only needed for
two purposes: To show that for every normalizing rewrite 
sequence there is also an innermost rewrite sequence,
and to prove that soundness for normalizing rewrite sequences
implies soundness for joinability.

Finally, adapting the result to
more complex transformations that have better properties
towards preserving confluence (in particular 
\emph{structure-preserving transformations}, most notably the transformations 
of \cite{ppdp03-antoy-brassel-hanus} and its extension to DCTRSs in 
\cite{wpte14}) might improve this result further. \\[1em]

\noindent \textbf{Acknowledgements:} I am grateful to the anonymous reviewers for their detailed comments on this paper and an earlier version of it.

\bibliographystyle{eptcs}
\bibliography{references}

\end{document}